\documentclass[journal,10pt]{IEEEtran}
\makeatletter
\newif\if@restonecol
\makeatother

\usepackage{bm}
\usepackage{graphicx}
\usepackage{subfigure}
\usepackage{epsfig}
\usepackage{float}
\usepackage{upgreek}
\usepackage{amssymb, amsmath, amsthm, amsfonts}
\usepackage{cite}
\usepackage{booktabs}
\usepackage{threeparttable}
\usepackage{multicol}
\usepackage{multirow}
\usepackage[ruled]{algorithm2e}

\usepackage{amsmath}

\newtheorem{lemma}{Lemma}
\begin{document}

\title{An Intelligent Prediction System for Mobile Source Localization Using Time Delay Measurements}
\author{Hengnian Qi, Xiaoping Wu, and Naixue Xiong
\thanks{This work was supported by Zhejiang Key R\&D Plan 2017C03047 and Zhejiang Province Key Laboratory of Smart Management \& Application of Modern Agricultural Resources under Grant 2020E10017. \emph{(Corresponding author: Xiaoping Wu)}}
\thanks{Hengnian Qi and Xiaoping Wu are with the School of Information Engineering, Huzhou University, Huzhou 313000, China. (Email: qhn@zjhu.edu.cn; wxp@zafu.edu.cn).}
\thanks{Naixue Xiong is with the Department of Mathematics and Computer Science, Northeastern State University, Tahlequah, OK, 74464, USA. (Email: xiongnaixue@gmail.com). }}
\maketitle

\begin{abstract}
In this paper, we introduce an intelligent prediction system for mobile source localization in industrial Internet of things. The position and velocity of mobile source are jointly predicted by using Time Delay (TD) measurements in the intelligent system. To predict the position and velocity, the Relaxed Semi-Definite Programming (RSDP) algorithm is firstly designed by dropping the rank-one constraint. However, dropping the rank-one constraint leads to produce a suboptimal solution. To improve the performance, we further put forward a Penalty Function Semi-Definite Programming (PF-SDP) method to obtain the rank-one solution of the optimization problem by introducing the penalty terms. Then an Adaptive Penalty Function Semi-Definite Programming (APF-SDP) algorithm is also proposed to avoid the excessive penalty by adaptively choosing the penalty coefficient. We conduct experiments in both a simulation environment and a real system to demonstrate the effectiveness of the proposed method. The results have demonstrated that the proposed  intelligent APF-SDP algorithm outperforms the PF-SDP in terms of the position and velocity estimation whether the noise level is large or not.
\end{abstract}
\IEEEpeerreviewmaketitle

\section{Introduction}\label{sec:introduction}
   Nowadays, a popular trend in many application systems is the using of mobile sources, such as Autonomous Vehicle (AV) and Unmanned Aerial Vehicle (UAV). The using of mobile sources shows their great advantages for the flexible mobility ability. Apparently, these mobile sources of AV and UAV can complete some complicated tasks, such as location-based services, radar or sonar navigation, target tracking, wireless network coverage and sensing enhancement, data collection, etc~\cite{TSMCS2015,TST2017,TSMC2018,ACCESS18,TII2019,ComNet2019}. It is very crucial to obtain the position parameter along with the velocity of mobile source when the mobile source assists to complete these tasks. Due to the mobility of mobile source, the position obtaining of mobile source is more complicated compared with immobile source~\cite{JS2010,IS2014,TNSE2020}. A common way to obtain the position is to utilize the sensors with known positions. Besides the sensors, the ranging information between source and sensor is also measured using some technologies, such as Time Of Arrival (TOA)~\cite{TOA1,TII19}, Time Difference Of Arrival (TDOA)~\cite{TDOA1,TDOA2}, Received Signal Strength (RSS)~\cite{RSS1,RSS2}, Angle Of Arrival (AOA)~\cite{AOA}, and recently popular Time Delay (TD)~\cite{TD1,TD2,TD3}. Among these technologies, TD is considered to be a simple and effective method to to predict the position and velocity of the mobile source.~\cite{CASSP2020}.

   Most of recent researches on the mobile source are focused on the obtaining of motion parameters by using the motion sensor or Doppler shift measurements~\cite{TWC2018,ISCIT2018,TSP2019}. In~\cite{IOT2018}, the inertial navigation method is proposed to predict the position of the mobile source using the motion sensors. When the motion data is subjected to error, the filtering method is proposed to improve the performance in~\cite{TVT2016,TII2018}. A linear algebraic method is put forward to predict the motion parameters using the measurement of differential delay and Doppler shift~\cite{SPL2016}. The direct acquisition of motion data depends on the hardware devices, and the measurements may exist errors.  Therefore, recent researches are concentrated on how to utilize various ranging methods to directly predict the motion parameters of mobile sources~\cite{RSS2}.

   In this paper we propose an intelligent practice prediction method of motion parameters including the position and velocity of mobile source by only using the time delay (TD) measurements. The proposed method does not require any motion sensors or Doppler shift measurements. To predict the position and velocity of the mobile source, the optimization model is firstly built by representing the measurement equation into a linear form. Then a Relaxed Semi-Definite Programming (RSDP) form is designed by relaxing the non-convex model into a convex problem. However, the solution to the RSDP problem is suboptimal due to the relaxation of rank-one constraint. To obtain the rank-one solution, we also propose the Penalty Function (PF) method by introducing the penalty terms. Then the APF-SDP algorithm is designed to obtain the optimal performance of the prediction problem by adaptively choosing the penalty coefficient.

   We also deploy a simulation environment and a real experiment system to fully demonstrate the effectiveness of the proposed algorithms. The real system is composed of a mobile car equipped with Ultrasonic module (UM) and nine sensors. Besides the UM, motion sensors are also equipped to the mobile car and used to measure the true position and velocity of the mobile source. The experimental results show that APF-SDP significantly outperforms the PF-SDP in terms of the position and velocity accuracy. 20\% of the position error is larger than 0.05 m for RSDP and 0.04 m for PF-SDP. However, it is reduced to near 0.01 m for the APF-SDP. By adaptively choosing the penalty coefficient, the performance of APF-SDP is still sufficiently close to the expected Cram\'{e}r-Rao Lower Bound (CRLB) of the prediction problem even if the number of sensors is varied from 5 to 9. It confirms the advantage of APF-SDP by achieving the rank-one constraint.

   Our proposed TD-based localization model does not require any motion sensors or Doppler shift measurements, so it differs from the existing mobile source localization problem in~\cite{SP2017,IOT2018}. Moreover, to obtain the well performance of the estimator, we propose the rank-one solution method using the penalty function. Our rank-one solution of AFP-SDP method is globally convergent by adaptively choosing the penalty coefficient and essentially different from the existing rank-one solution methods proposed in~\cite{FTO2016,TWC19,ICL2020}. The contributions of this work are summarized as follows,

   \begin{enumerate}
   \item To predict the position along with the velocity of the mobile source, dropping the rank-one constraint produces the RSDP problem. To design a tighter convex model, we propose the penalty function method to obtain a rank-one solution by introducing the penalty terms.
     \item To avoid the excessive penalty, we also put forward the APF-SDP algorithm by adaptively choosing the penalty coefficient. We have theoretically proven that the APF-SDP can provide optimal performance of the prediction problem by achieving the rank-one constraint.
     \item We have also developed an intelligent practice prediction system to demonstrate our proposed algorithms using TD measurements. The results from both the simulated and real experiments show that the APF-SDP algorithm provides better performance than the PF-SDP whether the noise level is large or not.
    \end{enumerate}
   The rest of this paper is structured as follows. Related works are introduced in Section~\ref{sec:works}. Section~\ref{sec:model} presents the system model and definitions. Section~\ref{sec:sdp} in detail describes our proposed intelligent prediction system. In Section~\ref{sec:complexity}, the computational complexity of these proposed algorithms is  derived. The numerical simulations and real experiments are analyzed in Section~\ref{sec:evaluations}. The conclusions and future work are presented in Section~\ref{sec:conclusion}. This paper contains a number of symbols. Following the convention, the matrix is represented by bold case letter. If the matrix  is denoted by $(\ast)$, $(\ast)^{-1}$ and $(\ast)^{T}$ indicate the matrix inverse and transpose operator, respectively. $\lVert \ast \rVert$ denotes  $\mathit \ell_2$ norm. $\bm A_{i,j}$ is the $(i,j)$th element of matrix $\bm A$. $\bm 0_{m}$ represents all-zero vector with the length $m$, and $\bm I_p$ and $\bm 0_{m\times n}$ are $m\times m$ identity and $m\times n$ zero matrices. Tr($\bm A$) and rank($\bm A$) stand for the trace and rank of $\bm A$, respectively. $\bm A\succeq 0$ indicates that $\bm A$ is positive semidefinite.

\section{Related Works}\label{sec:works}
  Based on the ranging information, various algorithms are proposed to predict the position and the velocity of mobile source. The popular algorithms include Maximum Likelihood Estimator (MLE)~\cite{ML1,ML2}, alternating direction method of multiplier (ADMM) method~\cite{ADMM1,ADMM2}, linear estimator~\cite{LS1,LS2}, and convex method~\cite{convex1,convex2}. The numerical solution of MLE requires a very good initial solution to guarantee its global convergence. The ADMM method provides the optimal estimate of source position by converting the unconstrained nonlinear problem into an equivalent constrained form. Due to the nonlinear nature of the optimization model, the MLE or ADMM method will be trapped in local optimum. To avoid the problem, the linear estimator directly represents the unknown variables as an algebraic analytic form by converting the nonlinear model into a linear problem. However, the constrained relationship among the variables is difficult to be exploited in the linear estimator, so the performance is not well enough. Convex method does not depend on the initialization for its global convergence and gradually becomes a popular method for the source position prediction problem. The convex methods can be achieved by Semi-Definite Programming (SDP)~\cite{SDP1,SDP2,SDP3} and Second Order Cone Programming (SOCP)~\cite{SOCP} which can be efficiently solved by using existing algorithms such as interior point methods~\cite{SJO1996}.

  A common method to obtain an SDP problem is to relax the non-convex optimization model into a convex form by dropping the rank-one constraint. The rank-one relaxation may lead to produce a suboptimal solution that is not the optimal solution of the original optimization problem. To obtain the rank-one solution of the convex SDP problem, many mathematical methods are proposed to deal with the troublesome problem~\cite{MP2005,OMS2006,NC2013,CCSE2020}. To solve the low-rank SDP problems, the factorization method is introduced by obtaining a reformulation of the original SDP problem in~\cite{SIAM2010}. A modified interior point method is proposed to solve low-rank SDP problem in~\cite{CDC2017}. Although these mathematical tools deal with the low-rank SDP problems, the solutions are also not guaranteed to be global convergence for their nonlinear or non-convex nature.

  To obtain the rank-one solution with global convergence, recently the two-stage method is proposed by refining the initial solution of the relaxed SDP problem. In~\cite{FTO2016}, a rank-reduction method is designed by solving an incremental matrix of the solution when an initial rank-maximum solution has been obtained with the relaxed SDP problem. The relaxed SDP problem provides a suboptimal solution, so a feasible method to solve the rank-one solution is to continuously tighten the relaxed problem. In~\cite{TWC19}, a set of SOCP constraints is added to tighten the convex model and find a rank-one solution of the original SDP problem. Above rank-one solutions strongly depend on the initial solutions of the relaxed SDP problem. Therefore, if the initial solutions are not accurate enough, the rank-one solution may fail. In~\cite{ICL2020}, the Penalty Function (PF) method is firstly proposed to obtain the rank-one solution of the SDP problem by introducing the penalty term. The rank-one solution of PF method provides a global optimum solution by choosing an appropriate penalty coefficient. However, too large penalty coefficient will lead to the occurrence of excessive penalty, which badly affects the performance of the solution.

  In this paper, we mainly propose an intelligent practice system, in which the motion parameters including the position and the velocity of the mobile source
   are predicted by only using the TD measurements. The proposed method does not require any motion sensors or Doppler shift measurements, and is thus applicable to track the mobile source, such as the AV or UAV.
   \section{System Model and Definitions}\label{sec:model}
  Assuming that $M$ sensors with known positions are placed in a $p$-dimensional ($p=2$ or $p=3$) scenario. The known positions of the sensors are denoted by $\bm s_{i}\in\mathbb{R}^p$, $i=1,2,\ldots,M$. In the same scenario, a mobile source starts from an initial position $\bm u\in\mathbb{R}^p$ at a constant velocity $\bm v\in\mathbb{R}^p$. Fig. 1 shows the diagram of TD measurement between mobile source and sensor. A mobile source transmits a signal at the initial position. Then the signal is received by the sensors and immediately reflected to the mobile source, thus forming a time delay that is formulated as:
   \begin{equation}
   \label{eq:01}
   t_{i}^o=\frac{1}{c}(\|\bm u-\bm s_i\|+\|\bm u_i-\bm s_i\|),
   \end{equation}
   where $t_{i}^o$ is the true time delay, $c$ is the propagation speed of the signal, $\bm u_i$ is the instantaneous position of the mobile source when the signal from the $i$th sensor is received by the mobile source. Apparently, we have
    \begin{equation}
   \label{eq:02}
   \bm u_i=\bm u+\bm v t_{i}^o.
   \end{equation}
   Substituting \eqref{eq:02} into \eqref{eq:01} and multiplying $c$ on the both sides produce
   \begin{equation}
   \label{eq:03}
   ct_{i}^o=\|\bm u-\bm s_i\|+\|\bm u+\bm v t_{i}^o-\bm s_i\|.
   \end{equation}
    Moving the term $\|\bm u-\bm s_i\|$ to the left side of \eqref{eq:03} and squaring on the both sides give
   \begin{equation}
   \label{eq:04}
   2(\bm u-\bm s_i)^T\bm v+2cd_i+(\bm v^T\bm v-c^2)t_i^o=0,
   \end{equation}
   where $\bm d_i=\|\bm u-\bm s_i\|$. Therefore, the true time delay is written as:
      \begin{equation}
   \label{eq:05}
   t_i^o=\frac{2(\bm u-\bm s_i)^T\bm v+2cd_i}{c^2-\bm v^T\bm v}.
   \end{equation}
   The measured time delay $t_i$ is always subject to error and given by
    \begin{equation}
   \label{eq:06}
   t_i=t_i^o+n_i,
   \end{equation}
   where $i=1,2,\ldots,M$, $n_i$ is the noise that constructs a vector form $\bm n=[n_1,n_2,\ldots,n_M]^T$. The noise vector $\bm n$ is always zero mean Gaussian distribution with covariance matrix $\bm Q_n$.
     \begin{figure}
  \centering
  \includegraphics[angle=0, width=0.25\textwidth]{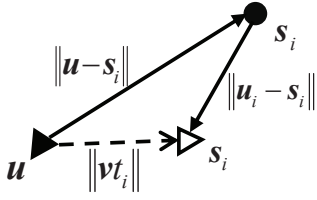}
  \caption{A diagram of TD measurements between sensor and mobile source.}
  \label{Fig1.lable}
  \end{figure}

  The goal of our system model is to predict the initial position and velocity of the mobile source using the noisy measurement $t_i$. In the following, the relaxed SDP (RSDP) method is firstly proposed to predict the position and velocity of the mobile source by dropping the rank-one constraint. The drop of the rank-one constraint is considered as a relaxation method that leads to produce an SDP solution with the rank higher than 1. So the performance of RSDP is suboptimal. To improve the performance, the penalty function method is proposed to obtain the rank-one solution by introducing the penalty terms.

 \section{Intelligent Prediction System}\label{sec:sdp}
  In this section, we in detail introduce the intelligent prediction system for predicting the position and velocity of mobile source using TD measurement.  The convex SDP problem shows the advantages for its global convergence and can be solved efficiently using interior-point algorithms. In  our proposed system, we mainly concentrate on the convex SDP solution to the mobile source localization problem. To obtain a convex SDP form, a general approach is to relax the non-convex optimization problem into a convex form, then the position and the velocity of mobile source are predicted by extracting from the SDP solution.
 \subsection{RSDP Algorithm}
   Substituting \eqref{eq:06} into \eqref{eq:04} yields the following expression:
    \begin{equation}
   \label{eq:07}
   -2\bm s_i^T\bm v+2\bm u^T\bm v+t_i\bm v^T\bm v+2cd_i-c^2t_i=\varepsilon_i,
   \end{equation}
   where $\varepsilon_i=(\bm v^T\bm v-c^2)n_i$, $i=1,2,\ldots,M$. To establish the optimization model, a new unknown vector $\bm x$ is defined by $\bm x=[\bm u^T,\bm v^T,\bm u^T\bm v,\bm v^T\bm v,\bm d^T,1]^T\in\mathbb{R}^{M+2p+3}$, $\bm d=[d_1,d_2,\ldots,d_M]^T$. Then \eqref{eq:07} is also represented by a linear matrix form:
    \begin{equation}
  \label{eq:08}
   \bm {Gx}=\bm\varepsilon,
  \end{equation}
  where $\bm\varepsilon=[\varepsilon_1,\varepsilon_2,\ldots,\varepsilon_M]^T$. According to the expression of \eqref{eq:07}, $\bm\varepsilon\in\mathbb{R}^{M}$ and $\bm G\in\mathbb{R}^{M\times(M+2p+3)}$ are also defined by
 \begin{subequations}
  \begin{align}
   \label{eq:09}
   & \bm\varepsilon=\bm B\bm n,  \\
   & \bm B=(\bm v^T\bm v-c^2)\bm I_M,\\
   &   \bm G=[\bm g_1^T,\bm g_2^T,\ldots,\bm g_M^T]^T, \\
   & \bm g_i=[\bm 0_{p}^T,-2\bm s_i^T,2,t_i,\bm 0_{i-1}^T,2c,\bm 0_{M-i}^T,-t_ic^2]^T,
   \end{align}
  \end{subequations}
  Therefore, the weighted least square (WLS) solution to the prediction problem is formulated as
\begin{subequations}
\begin{align}
 \label{eq:10}
 &  \min_{\bm x}\quad (\bm {Gx})^T\bm Q^{-1}(\bm {Gx})\nonumber\\
 &\mathrm{s.t.} \quad \|\bm x_{p+1:2p}\|^2=\bm x_{2p+2}, \\
 & \quad \quad \|\bm x_{1:p}-\bm s_i\|=\bm x_{2p+2+i},i=1,2\ldots,M, \\
 & \quad \quad \bm x_{1:p}^T\bm x_{p+1:2p}=\bm x_{2p+1}, \\
 & \quad \quad \bm x_{M+2p+3}=1,
\end{align}
\end{subequations}
 where $\bm Q\in\mathbb{R}^{M\times M}$ is the covariance matrix with respect to $\bm\varepsilon$, and $\bm Q$ is further obtained by
 \begin{align}
 \label{eq:11}
 \bm Q=\bm B\bm Q_n \bm B.
\end{align}
 The equality condition (10a) is also rewritten as
   \begin{equation}
   \label{eq:12}
  \|\bm D_0\bm x\|^2=\bm x_{2p+2},
   \end{equation}
  where $\bm D_0=[\bm 0_{p},\bm I_{p},\bm 0_{p\times(M+3)}]$. The equality condition (10b) is also given by
     \begin{equation}
   \label{eq:13}
  \|\bm D_i\bm x\|=\bm x_{2p+2+i},
   \end{equation}
   where $\bm D_i=[\bm I_p,\bm 0_{p\times(M+p+2)},\bm s_i]$, $i=1,2,\ldots,M$. When the equality conditions (10a) and (10b) are equivalent to the conditions (12) and (13), (10) is rewritten as
  \begin{align}
  \label{eq:14}
  &  \min_{\bm x}\quad (\bm {Gx})^T\bm Q^{-1}(\bm {Gx})\nonumber\\
  &\mathrm{s.t.} \quad \|\bm D_0\bm x\|^2=\bm x_{2p+2}, \nonumber\\
  & \quad \quad \|\bm D_i\bm x\|=\bm u_{2p+2+i},i=1,2\ldots,M, \nonumber\\
  & \quad \quad \bm x_{1:p}^T\bm x_{p+1:2p}=\bm x_{2p+1}, \nonumber\\
  & \quad \quad \bm x_{M+2p+3}=1.
  \end{align}

 To further express \eqref{eq:14} as a convex form, a new unknown matrix is defined by $\bm X=\bm x\bm x^T$. It is obviously shown that $\mathrm{rank(\bm X)}=1$. Thus, \eqref{eq:14} is given by
  \begin{subequations}
  \begin{align}
  \label{eq:15}
  &  \min_{\bm X}\quad \mathrm{Tr}(\bm C\bm X)\nonumber\\
  &\mathrm{s.t.} \quad \mathrm{Tr}(\bm D_0^T\bm D_0\bm X)=\bm X_{2p+2,2p+2},\\
  & \quad \quad  \mathrm{Tr}(\bm D_i^T\bm D_i\bm X)=\bm X_{2p+2+i,2p+2+i},i=1,2\ldots,M, \\
  & \quad \quad \sum\nolimits_{i=1}^p\bm X_{i,p+i}=\bm X_{2p+1,2p+1}, \\
  & \quad \quad \bm X_{M+2p+3,M+2p+3}=1,\bm X\succeq 0, \\
  & \quad \quad \mathrm{rank}(\bm X)=1,
  \end{align}
  \end{subequations}
  where $\bm C=\bm G^T\bm Q^{-1}\bm G$. Due to the rank-one constraint of (15e), the expression of (15) is non-convex. However, dropping the rank-one condition of $\bm X$ yields a convex SDP form,
  \begin{align}
  \label{eq:16}
  &  \min_{\bm X}\quad \mathrm{Tr}(\bm C\bm X) \nonumber\\
  &\mathrm{s.t.} \quad \bm A_i\bm X=b_i, i=1,\ldots,M+3,
  \end{align}
 where $b_i=[\bm 0_{1\times(M+2)},1]^T$, the sparse matrices $\bm A_i$ can be easily obtained by the expression of (15), $i=1,\ldots,M+3$. It is obviously shown that the problem depicted by \eqref{eq:16} is a typical SDP problem that can be effectively solved using a convex optimization package such as SEDUMI and SDPT3. Unfortunately, dropping the rank-one constraint relaxes the problem depicted by (15) and leads to produce a suboptimal SDP solution with a rank higher than 1, which implies that the solution of \eqref{eq:16} may not be the optimal solution of the original problem depicted by (15). To obtain the optimal performance, we propose the penalty function method to meet the rank-one constraint.
\subsection{PF-SDP Algorithm}
  The convex problem depicted by (16) provides a relaxed SDP solution to the mobile source localization problem. However, the relaxed SDP solution of (16) is suboptimal due to the drop of rank-one constraint. To obtain the optimal performance, we attempt to  find a rank-one solution by introducing the penalty terms. To design the PF-SDP method, we firstly prove the conclusion of Lemma 1.

\begin{lemma}
For a given solution $\bm X\in\mathbb{R}^{N\times N}$ to the problem depicted by (16), if $\bm X_{i,i}=\bm X_{i,N}^2$ $(N=M+2p+3$, $i=1,2,\ldots,N-1)$, then $\mathrm{rank}(\bm X)=1$.
 \end{lemma}
 \begin{proof}
 If $\bm X$ is a positive semidefinite solution of (16), its principal $2\times 2$ submatrix is also a positive semidefinite matrix. It is noted that $\bm X_{N,N}=1$. By selecting the $(i,i)$ and $(N,N)$ entry of $\bm X$ as the principal diagonal elements, a new positive semidefinite submatrix is constructed and given by
 \begin{equation}
 \left[\begin{array}{ccc}
  \bm X_{i,i}  & \bm X_{i,N}   \\
 \bm X_{N,i}   &  1  \\
 \end{array} \right]\succeq 0,
 \label{eq:17}
  \end{equation}
 where $i=1,2,\ldots,N-1$. Since $\bm X_{i,N}=\bm X_{N,i}$, we have $\bm X_{i,i}\geq\bm X_{i,N}^2$. When the equality condition is satisfied, it is also modified as
    \begin{equation}
    \sqrt{\bm X_{i,i}}=\bm X_{i,N}.
    \label{eq:18}
   \end{equation}
 Every principal $3\times 3$ submatrix of positive semidefinite matrix $\bm X$ is also positive semidefinite. Similarly, selecting the $(i,i)$, $(j,j)$, and $(N,N)$ entry of $\bm X$ as the principal diagonal elements and using the equality (18), we also construct a $3\times 3$ submatrix that is given by
   \begin{equation}
    \left[\begin{array}{cccc}
   \bm X_{i,i}  & \bm X_{i,j}  & \sqrt{\bm X_{i,i}}\\
   \bm X_{j,i}  & \bm X_{j,j}  & \sqrt{\bm X_{j,j}}\\
   \sqrt{\bm X_{i,i}}  & \sqrt{\bm X_{j,j}}  & 1
   \end{array} \right]\succeq 0,
    \label{eq:19}
   \end{equation}
 where $i,j=1,2,\ldots,N-1$, $i<j$. \eqref{eq:18} is also equivalent to the following expression:
    \begin{equation}
    \left[\begin{array}{ccc}
   0  & \bm X_{i,j}-\sqrt{\bm X_{i,i}\bm X_{j,j}}  \\
   \bm X_{i,j}-\sqrt{\bm X_{i,i}\bm X_{j,j}}  & 0  \\
   \end{array} \right]\succeq 0,
    \label{eq:20}
   \end{equation}
  where $i,j=1,2,\ldots,N-1$, $i<j$. The equation (20) holds if and only if $\bm X_{i,j}=\sqrt{\bm X_{i,i}\bm
  X_{j,j}}$ for all $i,j=1,2,\ldots,N-1$, $i<j$. Therefore, the positive semidefinite matrix $\bm X$ is also rewritten as
  \begin{small}
  \begin{equation}
    \bm X=
   \left[\begin{array}{cccccccc}
   \bm X_{1,1}  & \sqrt{\bm X_{1,1}\bm X_{2,2}} &\ldots &\sqrt{\bm X_{1,1}}\\
   \sqrt{\bm X_{1,1}\bm X_{2,2}}& \bm X_{2,2}  & \ldots & \sqrt{\bm X_{2,2}}\\
     \vdots   & \vdots & \vdots & \vdots\\
   \sqrt{\bm X_{1,1}}& \sqrt{\bm X_{2,2}} & \ldots & 1
    \end{array} \right].\\
    \label{eq:21}
   \end{equation}
 \end{small}
  It is obviously shown that $\mathrm{rank}(\bm X)=1$.
 \end{proof}
 To achieve the equality (18), a feasible method is to restrict the magnitude of $\bm X_{i,i}$ by introducing the penalty terms. Therefore, a new PF-based SDP form is given by
 \begin{align}
 \label{eq:22}
 &  \min_{\bm X}\quad \mathrm{Tr}(\bm C\bm X)+\eta\sum\nolimits_{i=1}^{N-1}\bm X_{i,i}\nonumber \\
 & \mathrm{s.t.} \quad \bm A_i\bm X=b_i, i=1,\ldots,M+3,
 \end{align}
 where $\eta$ is a strictly positive constant and called as penalty coefficient. When $\eta$ is gradually increased, $\bm X_{i,i}$ will sufficiently close to $\bm X_{i,N}^2$. Eventually, large enough $\eta$ will lead that $\mathrm{rank}(\bm X)=1$.

 When $\eta$ is chosen to be large enough, the SDP problem depicted by \eqref{eq:22} provides a rank-one solution of the positive semidefinite $\bm X$. However, a concerned issue is that whether the rank-one solution is the original solution of (10) or not. It is directly shown that the problem depicted by (10) has only $M+2p+3$ variables. Since there are $M+3$ constraints among the variables in defined vector $\bm x$, so the problem has only $2p$ independent variables. It is also observed that the dimension of $\bm X$ in the problem depicted by \eqref{eq:22} is also $M+2p+3$ when $\mathrm{rank}(\bm X)=1$. Due to the similar constraints, the problem depicted by \eqref{eq:22} has also only $2p$ independent variables. Therefore, the rank-one solution of the SDP problem depicted by \eqref{eq:22} is also the original solution of (10) since the convex SDP form of \eqref{eq:22} is derived from the non-convex problem depicted by (10).
 \begin{figure}
\centering
\includegraphics[angle=0, width=0.50\textwidth]{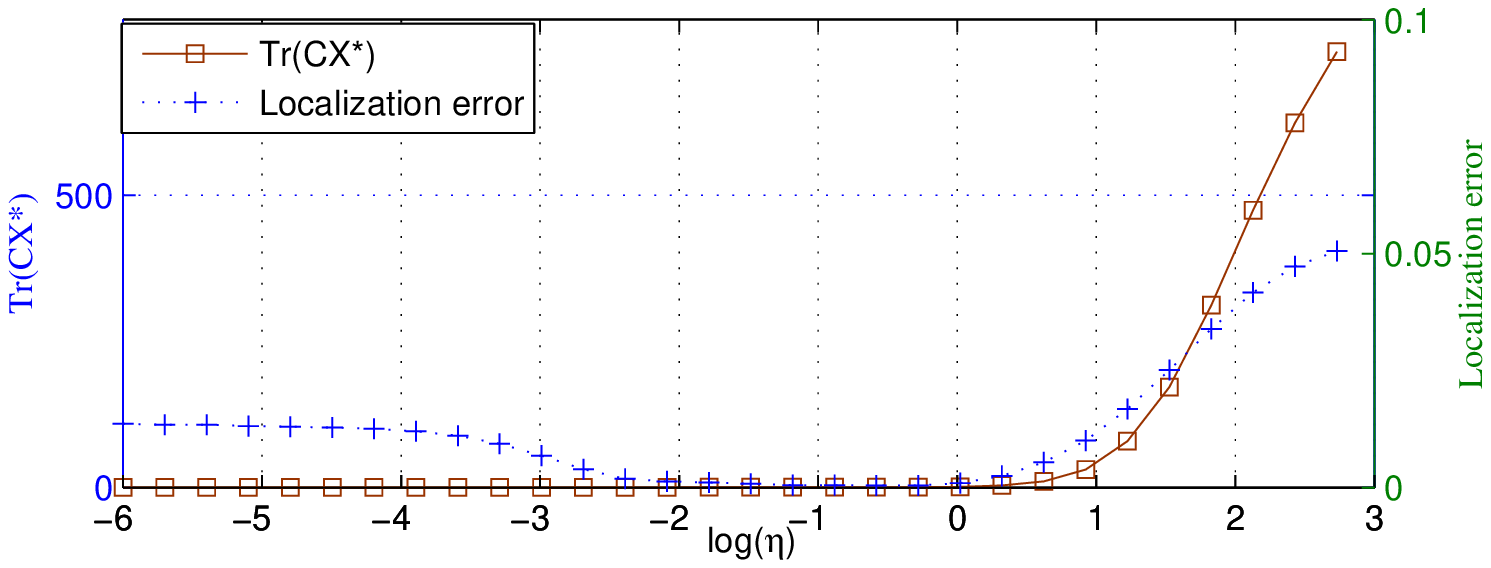}
\caption{Localization error and optimal objective $\mathrm{Tr}(\bm C\bm X^{*})$ with a test.}
\label{Fig2.lable}
\end{figure}

 An appropriate penalty coefficient can ensure the effective working of penalty terms. Although large enough $\eta$ will achieve the rank of the solved $\bm X$ to be one, too large $\eta$ may lead to a risk of excessive penalty, which will badly affect the predicted result. Hence, the choosing of $\eta$ is crucial to obtain the well performance of the prediction problem. To choose an appropriate penalty coefficient, we propose an adaptive penalty function-based SDP (AFP-SDP) algorithm to obtain a rank-one solution for the SDP optimization problem in the following.
\subsection{APF-SDP Algorithm}\label{sec:APF-SDP}
 Too large $\eta$ increases the proportion of penalty terms to the total cost function of (22), so it will weaken the original objective $\mathrm{Tr}(\bm C\bm X)$. To avoid the chosen $\eta$ to be too large, $\eta$ starts from a small value and increases gradually, until the rank-one condition is satisfied.
 To clearly observe the effect of the penalty coefficient in the problem depicted by (22), we conducted a random test and the results are plotted in Fig.1, where $\bm X^{*}$ and $\mathrm{Tr}(\bm C\bm X^{*})$ represent the optimal solution and objective, respectively. In the test, $\mathrm{log}\eta$ is gradually increased from -6 to 3 (i.e. $\eta$ is gradually increased from $10^{-6}$ to $10^{3}$). As can be seen that the penalty terms almost do not work when $\mathrm{log}\eta$ is increased from -6 to -4. The localization error is gradually reduced for the working of penalty terms when $\mathrm{log}\eta$ is increased from -4 to -2. The optimal performance is achieved when $\mathrm{log}\eta$ is set to the range of (-2, 0). However, if $\mathrm{log}\eta$ is continuously to be increased, the performance will become worse. The localization error will sharply increased due to the occurrence of excessive penalty when $\mathrm{log}\eta$ is larger than 1. The occurrence of excessive penalty badly affects the localization result, so it is crucial to detect whether the penalty is excessive or not. It can be seen from Fig.2 that the optimal objective $\mathrm{Tr}(\bm C\bm X^{*})$ is simultaneously increased with the localization error when the excessive penalty occurs. Therefore, the optimal objective $\mathrm{Tr}(\bm C\bm X^{*})$ can be used to detect the occurrence of excessive penalty.
 \begin{lemma}
 If $\bm X^{*}$ is an optimal rank-one solution of SDP problem depicted by (16), then $\mathrm{Tr}(\bm C\bm X^{*})\thicksim\sum_{i=1}^{M} \lambda_i\chi^2(1)$.
 \end{lemma}
 \begin{proof} If $\bm X^{*}$ is an optimal rank-one solution of SDP problem depicted by (16), a new vector $\bm x^{*}$ is defined by $\bm X^{*}=\bm x^{*}\bm x^{*T}$, where $\bm x^{*}\in \mathbb{R}^{M+2p+3}$. Since the problem is derived from (10), $\bm x^{*}$ is also an optimal solution of problem depicted by (10). Then it yields
  \begin{align}
   \label{eq:23}
   \mathrm{Tr}(\bm C\bm X^{*}) = \bm x^{*T}\bm C \bm x^{*}.
   \end{align}
 The prediction error of $\bm x^{*}$  is denoted by $\triangle\bm x^{*}$, then $\bm x^{*}+\triangle\bm x^{*}=\bm x^{o}$, where  $\bm x^{o}$ represents the true value of the defined $\bm x$.

 The constrained optimization problem depicted by (10) has only $2p$ independent variables that constructs a vector $\bm z=[\bm u^T, \bm v^T]^T$. If the prediction error of $\bm z$ is denoted as $\triangle\bm z$, we can obtain
 \begin{align}
  \label{eq:24}
  \frac{\partial (\bm G\bm x)}{\partial \bm x}\frac{\partial \bm x}{\partial \bm z}\triangle \bm z=\bm \varepsilon,
  \end{align}
  where $ \frac{\partial (\bm G\bm x)}{\partial \bm x}=\bm G$, $\frac{\partial \bm x}{\partial \bm z}$ is denoted as $\bm H$ and given by
 \begin{align}
\bm H=\left[
                 \begin{array}{ccccccc}
                   \bm I_p & \bm 0_{p\times p} & \bm v & \bm 0_{p} &\frac{\partial \bm d}{\partial\bm u} &\bm 0_{p}\\
                   \bm 0_{p\times p} & \bm I_p & \bm u & 2\bm v & \bm 0_{p\times M} & \bm 0_{p}\\
               \end{array}
               \right]^T,
 \label{eq:25}
 \end{align}
 where $\frac{\partial \bm d}{\partial\bm u}=[\frac{\bm u-\bm s_1}{\|\bm u\|},\frac{\bm u-\bm s_2}{\|\bm u\|},\ldots, \frac{\bm u-\bm s_M}{\|\bm u\|}]$. Therefore, \eqref{eq:24} is rewritten as
 \begin{align}
  \label{eq:26}
  \bm P\triangle \bm z=\bm \varepsilon,
  \end{align}
  where $\bm P=\bm G\bm H$. Since $\bm \varepsilon=\bm B\bm n$, the WLS solution to \eqref{eq:26} is
   \begin{align}
    \label{eq:27}
    \triangle \bm z=(\bm P^T\bm W\bm P)^{-1}\bm P^T\bm W \bm B\bm n,
    \end{align}
   where $\bm W=\bm Q^{-1}$. Since $\triangle\bm x^{*}=\bm H \triangle\bm z$, the prediction error $\triangle\bm x^{*}$  is obtained by
  \begin{align}
   \label{eq:28}
   \triangle\bm x^{*}=\bm E\bm n,
   \end{align}
   where $\bm E=\bm H(\bm P^T\bm W\bm P)^{-1}\bm P^T\bm W \bm B$. Since $\bm x^{*}=\bm x^{o}-\triangle\bm x^{*}$, the optimal objective is approximately given by
   \begin{align}
   \label{eq:29}
   \bm x^{*T}\bm C \bm x^{*}\approx \bm x^{*T}\bm C \bm x^{o}.
   \end{align}
   Substituting $\bm x^{*}=\bm x^{o}-\triangle\bm x^{*}$ into the right side of \eqref{eq:29} yields
       \begin{align}
   \label{eq:30}
   \bm x^{*T}\bm C \bm x^{*}=\bm x^{oT}\bm C \bm x^{o}- \bm x^{oT}\bm C \triangle \bm x^{*},
   \end{align}
   where $\bm C$ is interfered by the noise. Let $\bm C=\bm C^{o}+\triangle \bm C$, $\bm C^{o}$ and $\triangle \bm C$  denote the true value and the error term caused by noise, respectively. Similarly, we also define $\bm G=\bm G^{o}+\triangle \bm G$ in which $\bm G^{o}$
   and $\triangle \bm G$ represent the true value and the error term,
    \begin{subequations}
  \begin{align}
   \label{eq:31}
   &\triangle\bm G=[\triangle \bm g_1^T,\triangle \bm g_2^T,\dots,\triangle \bm g_M^T]^T  \\
   &\triangle \bm g_i=[\bm 0_{2p+1}^T,n_i,\bm 0_{M}^T,-c^2n_i]^T.
   \end{align}
  \end{subequations}
   According to the definition of $\bm C=\bm G^T\bm W\bm G$, we can obtain the true value $\bm C^{o}$ and the error term $\triangle \bm C$,
  \begin{subequations}
  \begin{align}
   \label{eq:32}
   &\bm C^{o}=\bm G^{oT}\bm W\bm G^{o}              \\
   &\triangle\bm C=\bm G^{oT}\bm W\triangle\bm G+ \triangle\bm G^{T}\bm W\bm G^{o}+\triangle\bm G^{T}\bm W\triangle\bm G. \end{align}
  \end{subequations}
   Since $\bm G^{o}\bm x^{o}=\bm 0_M$, the first term of the right side in \eqref{eq:30} is also further written as
    \begin{align}
   \label{eq:33}
   \bm x^{oT}\bm C \bm x^{o}\approx(\triangle\bm G\bm x^{o})^{T}\bm W(\triangle\bm G \bm x^{o}),
   \end{align}
   where $\triangle\bm G\bm x^{o}$ is also given by
   \begin{subequations}
  \begin{align}
   \label{eq:34}
   & \triangle\bm G\bm x^{o}=\bm F \bm n      \\
   &\bm F=\mathrm{diag}(\bm T\bm x^o)   \\
    &  \bm T=[\bm t_1^T,\bm t_2^T,\ldots,\bm t_M^T]^T \\
     &  \bm t_i=[\bm 0_{2p+1}^T,1,\bm 0_{M}^T,-c^2]^T.
   \end{align}
  \end{subequations}
   Therefore, \eqref{eq:33} is also rewritten as
    \begin{align}
   \label{eq:35}
   \bm x^{oT}\bm C \bm x^{o}=\bm n^T\bm F^T\bm W\bm F\bm n.
   \end{align}

   Similarly, the second term of the right side in \eqref{eq:30} can also rewritten as
   \begin{align}
   \label{eq:36}
      \bm x^{oT}\bm C \triangle\bm x^{*}=\bm x^{oT}(\triangle\bm G^{T}\bm W\bm G^{o}+\triangle\bm G^{T}\bm W\triangle\bm G)\triangle\bm x^{*}.
   \end{align}
   Neglecting the high order term of \eqref{eq:36} also yields
   \begin{align}
   \label{eq:37}
      \bm x^{oT}\bm C \triangle\bm x^{*}\approx(\triangle\bm G\bm x^{o})^T\bm W\bm G^{o}\triangle\bm x^{*}.
   \end{align}
   Using the expression of (28) and (34a), we can rewrite \eqref{eq:37} as
     \begin{align}
   \label{eq:38}
      \bm x^{oT}\bm C \triangle\bm x^{*}=\bm n^T\bm F^T\bm W\bm G^{o}\bm E\bm n.
   \end{align}
   According to the expressions \eqref{eq:35} and \eqref{eq:38}, \eqref{eq:30} is also rewritten as
  \begin{align}
  \label{eq:39}
  \bm x^{*T}\bm C \bm x^{*}\approx\bm n^T\bm F^T\bm W(\bm F-\bm G\bm E)\bm n,
  \end{align}
 where $\bm G^o$ is approximately equal to $\bm G$. Let $\bm L=\bm F^T\bm W(\bm F-\bm G\bm E)$. Therefore, $\bm x^{*T}\bm C \bm x^{*}$ conforms to the chi-square distribution $\sum_{i=1}^{M} \lambda_i\chi^2(1)$, which is also denoted by
  \begin{align}
   \label{eq:40}
   \bm x^{*T}\bm C \bm x^{*}\thicksim\sum_{i=1}^{M} \lambda_i\chi^2(1),
   \end{align}
  where $\lambda_i$ is the eigenvalue of the matrix $\bm L\bm Q_n$, $i=1,2,\ldots,M$.
  \end{proof}

  When $\eta$ is chosen to be too small, the rank of $\bm X^{*}$ may be larger than 1. The increasing of $\eta$ can ensure that the rank of $\bm X^{*}$ is gradually close to be one. To evaluate the rank of $\bm X^{*}$, we firstly propose an eigenvalue method according to the following threshold value,
  \begin{align}
 \label{eq:41}
 \tau=\frac{\widetilde{\lambda}_{N-1}}{\widetilde{\lambda}_{N}},
 \end{align}
 where $\widetilde{\lambda}_{N-1}$ and $\widetilde{\lambda}_{N}$ are the $N-1$th and the $N$th eigenvalue of $\bm X^{*}$, $N=M+2p+3$. It is obviously shown that the rank of $\bm X^{*}$ tends to be one when $\tau$ is sufficiently close to zero.

 To ensure the well performance of the PF-SDP, it is crucial to choose an optimal penalty coefficient which depends on various factors including the positions of sensors and mobile source, the number of sensors, and the noise level. Therefore, the optimal $\eta$ is not invariable under the different situations.
 In Algorithm 1, we propose an Adaptive Penalty Function-based SDP (APF-SDP) algorithm by adaptively choosing an optimal penalty coefficient.

 In Algorithm 1, two conditions is required to be judged. One is the condition $\tau<\delta$ ($\delta$ is a constant that is sufficiently close to zero), which is considered as rank-one condition and used to judge whether the rank of $\bm X^{*}$ is one or not. The other is the condition $\bm C\bm X<\epsilon$ ($\epsilon$ is also a threshold value determined by the chi-square distribution of $\bm C\bm X^{*}$), which is used to detect whether the penalty is excessive or not. $\eta$ starts from a small value and increases gradually with $\eta=\gamma\eta$ ($\gamma>1$), where $\gamma$ is called as step length.  When $\tau<\delta$ is satisfied, it is considered to meet the rank-one condition. Then using the condition $\bm C\bm X<\epsilon$, we further detect whether the penalty is excessive or not. When meeting the condition $\bm C\bm X<\epsilon$, we accept the solution that is considered as a final solution $\bm X^{*}$.
 If not, we relax the rank-one condition with $\delta=\alpha\delta$ ($\alpha>1$) and resolve the problem depicted by (22) until two conditions hold.

  \begin{algorithm}[t]
 \caption{APF-SDP Algorithm}
\LinesNumbered
\KwIn{$\bm A$, $\bm b$, $\bm C$, $\eta$, $\alpha>1$, $\gamma>1$, $\delta>0$, $\epsilon>0$}
\KwOut{$\bm X^{*}$}
\While{1}{
　　Using (17) to calculate the optimal $\bm X$\;
　　\eIf{$\tau<\delta$}{
　　　　\eIf{$\mathrm{Tr}(\bm C\bm X)<\epsilon$}{
　　　　$\bm X^{*}=\bm X$\;
        break;
　　}{
       $\delta=\alpha\delta$\;
        goto 2;
　　　}}{
　　　　$\eta=\gamma\eta$;
　　}
}
return $\bm X^{*}$
\end{algorithm}

\begin{table*}
    \centering
    \fontsize{8}{11}\selectfont
    \caption{Positions of ten sensors in simulations(km)}
    \begin{tabular}{c|c|c|c|c|c|c|c|c|c|c}
    \toprule
    \hline
     $x$ & 0.145 & -0.020 & -0.207 & 0.400 & -0.358 & 0.187 & -0.604 & 0.621 & -0.099 & 0.205\\\hline
     $y$ & -0.385 & 0.199 & -0.163 & -0.464 & 0.765 & -0.167 & -0.473 & 0.409  & 0.736& -0.456 \\\hline
    \bottomrule
    \end{tabular}\vspace{0cm}
    \label{tab:1}
    \end{table*}

\section{Computational Complexity}\label{sec:complexity}
  The worst-case complexity of solving an SDP problem in each iteration is $O(m^2n^2)$~\cite{CO2004,CUP2009}, where $m$ is the number of the equality constraints, $n$ is the dimension of the SDP cone. The number of iteration is bounded by  $O(\sqrt n\mathrm{ln}(1/\varsigma))$, where $\varsigma$ is the solution precision, $\sqrt n$ is the number of iterations caused by barrier parameter. For the proposed PF-SDP, we have $m=M+3$ and $n=M+2p+3$. Therefore, the computational complexity of the PF-SDP in (22) is given by
    \begin{equation}
    \label{eq:42}
    \text{PF-SDP Complexity}\simeq O\big((M+3)^2N^{2.5}\mathrm{ln}(1/\varsigma))\big),
    \end{equation}
   where $N=(M+2p+3)$. To choose an appropriate penalty coefficient, the APF-SDP requires $\kappa$ PF-SDP solutions, where $\kappa$ depends on $\eta_0$, the initial penalty coefficient, $\eta^{*}$, the optimal penalty coefficient, and the step length $\gamma$. Therefore, it needs to be satisfied that  $\eta_0\gamma^{\kappa}\geq\eta^{*}$. We can further obtain
   \begin{equation}
  \label{eq:43}
  \kappa\geq \frac{\mathrm{ln}\eta^{*}-\mathrm{ln}{\eta^0}}{\mathrm{ln}{\gamma}}.
  \end{equation}
  The complexity of APF-SDP in Algorithm 1 is $\kappa$ times of that of the PF-SDP and given by
  \begin{equation}
  \label{eq:44}
  \text{PF-SDP Complexity}\simeq O\big(\kappa(M+3)^2N^{2.5}\mathrm{ln}(1/\varsigma))\big).
  \end{equation}

  \section{Performance Analysis}\label{sec:evaluations}

  The RSDP, PF-SDP and APF-SDP algorithms are proposed to predict the initial position and velocity of the mobile source using the time delay measurements.
  To evaluate the performance of these proposed algorithms, the simulations were firstly conducted in a 2-D scenario. We have also done the simulations for the 3-D case and the observations were similar. The positions of ten sensors are randomly generated according to the uniform distribution $\mathcal{U}(-1,1)$ and listed in Tab.~\ref{tab:1}. In the same region, a mobile source starts from the initial position (0.2, -0.4) km with a constant velocity (-1, 1) m/s.  The propagation speed $c$ is randomly drawn from the range $(0.3, 0.4)$ km/s. TD measurements are generated based on the model equation (1), and the noise in (6) is  modeled as zero mean Gaussian distribution with covariance $\bm Q_n=\sigma^2\bm I_M$. The performance is evaluated using Mean Square Error (MSE) defined by
  \begin{equation}
    \label{eq:45}
     \left\{
      \begin{array}{ll}
     \mathrm{MSE}(\bm u)=\frac{1}{K}\sum_{k=1}^{K}\|\bm u_k-\bm u^o\|^2\\
     \mathrm{MSE}(\bm v)=\frac{1}{K}\sum_{k=1}^{K}\|\bm v_k-\bm v^o\|^2,
     \end{array}
     \right.
      \end{equation}
 where $\bm u^o$ and $\bm v^o$ represent the true initial position and velocity, $\bm u_k$ and $\bm v_k$ are the predicted results in the $k$th Monte Carlo (MC) run, $K$ is the number of MC runs ($K$ is set to 1000 in our simulations). The proposed algorithms are also compared with the Cram\'{e}r-Rao Lower Bound (CRLB) of the prediction problem.
  \subsection{Penalty Coefficient}
  The increasing of $\eta$ ensures that the rank of the SDP solution gradually tends to be one. However, the performance of the rank-one solution may be not optimal due to the occurrence of excessive penalty. Therefore, it is required to ensure that the rank-one solution has no the excessive penalty. To evaluate the improved performance of PF-SDP compared with the RSDP, the Logarithmic MSE difference (LOGMSED) is defined by
  \begin{equation}
  \label{eq:46}
  \mathrm{LOGMSED}=\mathrm{10log(MSE})_{\mathrm{PF}}-\mathrm{10log(MSE})_{\mathrm{R}},
  \end{equation}
  where $\mathrm{10log(MSE})_{\mathrm{PF}}$ and $\mathrm{10log(MSE})_{\mathrm{R}}$ represent the MSE in log-scale of PF-SDP and RSDP, respectively.

 The first eight sensors listed in Tab.~\ref{tab:1} are used to determine the initial position and velocity of the mobile source. The noise level $10\mathrm{log}\sigma^2$ is set to -40, -20, 0, respectively. Fig.~\ref{Fig3.sub.1} illustrates the LOGMSED($\bm u$) performance when the penalty coefficient $\eta$ is increased from $10^{-6}$ to $10^{3}$ (i.e. log$\eta$ is increased from -6 to 3). The penalty terms do not basically work when $\eta$ is too small. As can be seen that LOGMSED($\bm u$) is near zero when log$\eta$ is set to -6. When $\eta$ gradually increases, LOGMSED($\bm u$) becomes a negative value, which illustrates that the MSE of PF-SDP is smaller than that of the RSDP due to the working of penalty terms. When log$\eta$ is set to (0, 2) for $10\mathrm{log}\sigma^2=-40$, (-1, 1) for $10\mathrm{log}\sigma^2 = -20$, and  (-2, 0) for $10\mathrm{log}\sigma^2 = 0$, LOGMSED($\bm u$) reaches the least value, indicating the optimal performance of PF-SDP. Therefore, the optimal range of $\eta$ is variable due to the different noise level. When log$\eta$ is larger than the optimal range, LOGMSED($\bm u$) will be sharply increased, which confirms the occurrence of excessive penalty. For instance, when log$\eta$ is set to (2,3), LOGMSED($\bm u$) becomes a positive value for $10\mathrm{log}\sigma^2 = 0$, illustrating that the performance of PF-SDP is worse than that of the RSDP.

 When the parameter setup is the same with that of Fig.~\ref{Fig3.sub.1}, Fig.~\ref{Fig3.sub.2} shows the LOGMSED($\bm v$) with the increasing of $\eta$.
 It is also observed that the LOGMSED($\bm v$) becomes a negative value when log$\eta$ is gradually increased from -6. Then the LOGMSED($\bm v$) reaches a least value that manifests the optimal performance of PF-SDP. For instance, when $10\mathrm{log}\sigma^2$ is set to -40, the optimal log$\eta$ is set to the range of (0, 2), where the PF-SDP provides better performance and has almost 13 dB bias compared with the RSDP. When log$\eta$ is larger than the optimal range, the LOGMSED($\bm v$) will be sharply increased due to the occurrence of excessive penalty.

  \begin{figure}
 \centering
 \subfigure[logarithmic MSE difference of position.]{
 \label{Fig3.sub.1}
 \includegraphics[width=0.48\textwidth]{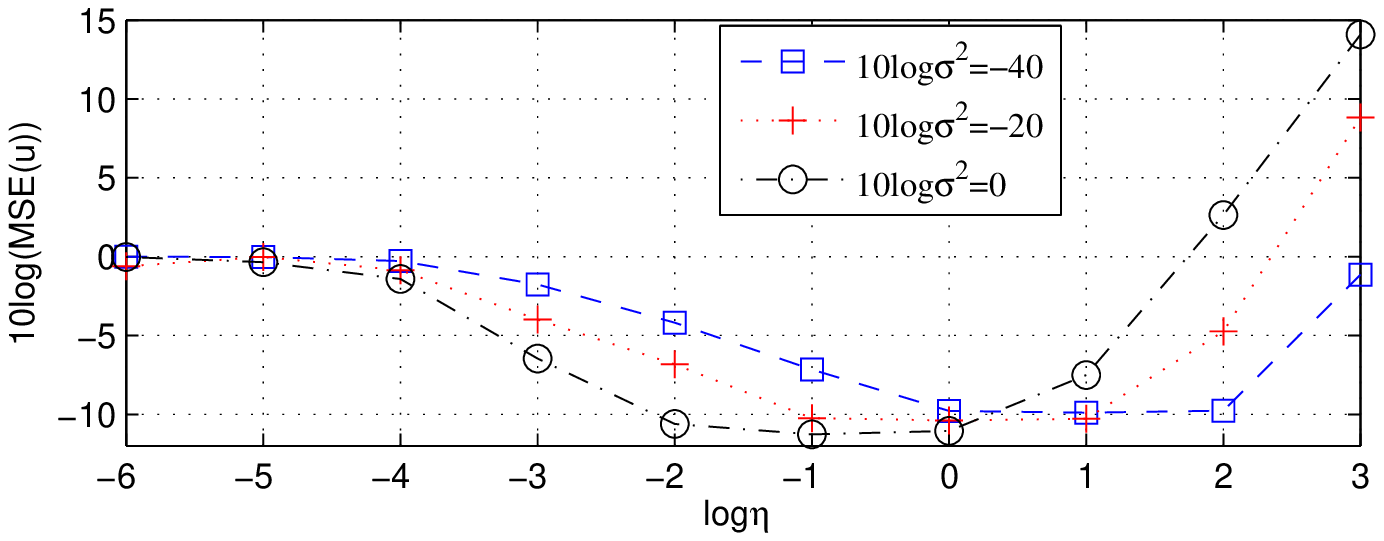}}
 \subfigure[logarithmic MSE difference of velocity.]{
 \label{Fig3.sub.2}
 \includegraphics[width=0.48\textwidth]{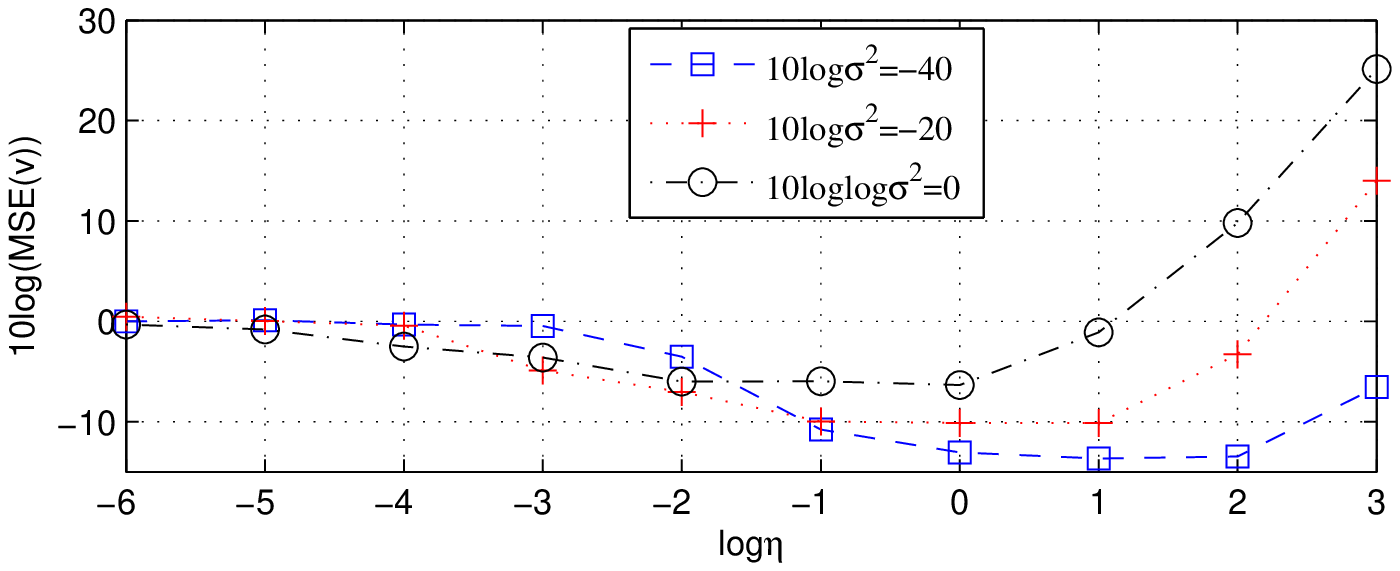}}
 \caption{Performance comparison as the $\eta$ varies.}
 \label{Fig3.lable}
 \end{figure}

  \begin{table}
    \centering
    \fontsize{8}{11}\selectfont
    \caption{Number of rank-one solutions when $\eta$ varies(1000 MC runs for each value of $\eta$, $10\mathrm{log}\sigma^2=0$)}
     \begin{tabular}{ccccccccccc}
    \toprule
    \hline
   {log$\eta$}  &-6  &-5    &-4   &-3   &-2  &-1  &0  &1 & 2\cr\hline
     Number & 0 & 3& 4& 12 & 28 & 306 & 825  & 1000 & 1000\cr\hline
     \bottomrule
    \end{tabular}\vspace{0cm}
    \label{tab:2}
    \end{table}

     \begin{table*}
    \centering
    \fontsize{8}{11}\selectfont
    \caption{Number of rank-one solutions when $\sigma^2$ varies (1000 MC runs for each value of $10\mathrm{log}\sigma^2$)}
      \begin{tabular}{cccccccccccccc}
    \toprule
    \hline
    10log$\sigma^2$&-60&-50&-40&-30&-20&-10&0&10&20\cr\hline
    Number ($\mathrm{log}{\eta}=-1$) & 0  & 1 & 9  &37 & 69 & 169 & 293 & 525 & 739 \cr\hline
    Number ($\mathrm{log}{\eta}=0$) & 2  & 34 & 112  & 354 & 531 & 619 & 834 & 999 & 1000 \cr\hline
    Number ($\mathrm{log}{\eta}=1$) & 205  & 542 & 675  & 769 & 872 & 999& 1000 & 1000 & 1000 \cr\hline
    \bottomrule
    \end{tabular}\vspace{0cm}
    \label{tab:3}
    \end{table*}
     \begin{figure}
  \centering
  \subfigure[Performance of predicted position.]{
  \label{Fig4.sub.1}
  \includegraphics[width=0.48\textwidth]{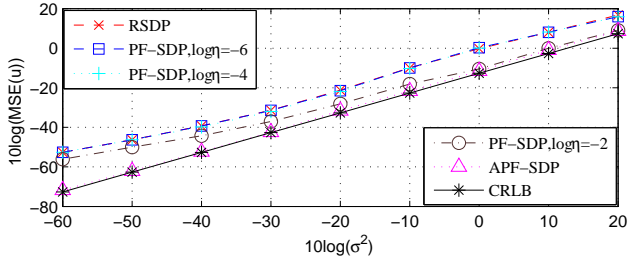}}
   \subfigure[Performance of predicted velocity.]{
  \label{Fig4.sub.2}
  \includegraphics[width=0.48\textwidth]{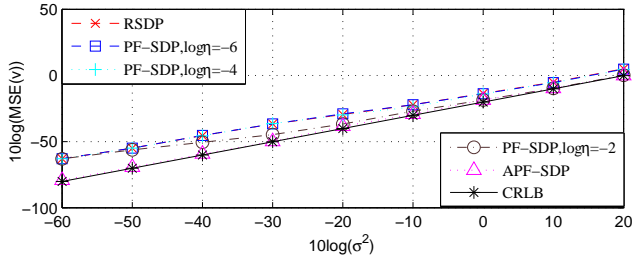}}
   \caption{Performance comparison as $\sigma^2$ varies.}
  \label{Fig4.lable}
  \end{figure}

    The rank of PF-SDP solution will gradually tend to be one as $\eta$ increases. If meeting $\tau<\delta$, it is considered as a rank-one solution.
    When 10log$\sigma^2$ is set to 0 dB, the number of rank-one solutions is shown in Tab.~\ref{tab:2}, where $\delta$ is set to $10^{-5}$. As can be seen that the number of rank-one solutions is increased with log$\eta$, since larger $\eta$ provides tighter SDP cone constraints. When log$\eta$ is larger than 1, all 1000 MC runs of PF-SDP solutions meet the rank-one condition. When 10log$\sigma^2$ is increased from -60 dB to 20 dB, the number of rank-one solutions is also illustrated in Tab.~\ref{tab:3}, where log$\eta$ is set to -1, 0, and 1, respectively. The number of rank-one solutions is also dramatically increased with 10log$\sigma^2$. When 10log$\sigma^2$ is set to -60 dB, the number of rank-one solutions is 2 for log$\eta = 0$, 205 for log$\eta = 1$. However, all 1000 MC runs meet the rank-one condition for log$\eta = 0$ and log$\eta = 1$ when 10log$\sigma^2$ is increased to 20 dB.
   \subsection{Performance with Simulations}
    \begin{figure}
 \centering
 \subfigure[Performance of predicted position.]{
\label{Fig5.sub.1}
\includegraphics[width=0.48\textwidth]{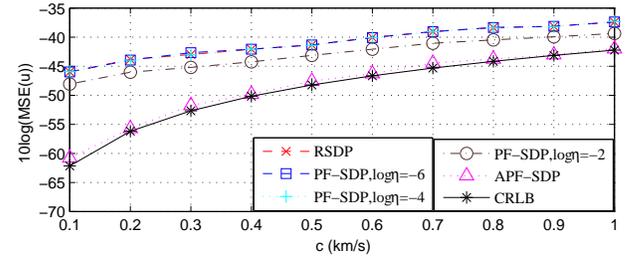}}
\subfigure[Performance of predicted velocity.]{
\label{Fig5.sub.2}
\includegraphics[width=0.48\textwidth]{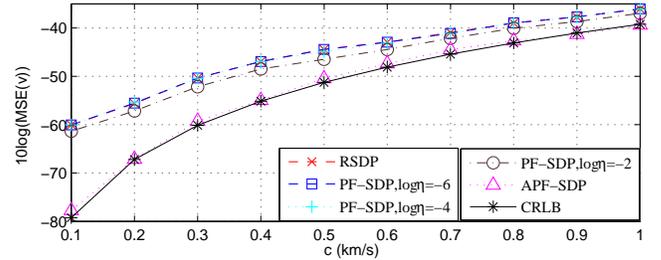}}
\caption{Performance comparison as $c$ varies.}
\label{Fig5.lable}
\end{figure}
\begin{figure}
 \centering
\subfigure[Performance of predicted position.]{
\label{Fig6.sub.1}
\includegraphics[width=0.48\textwidth]{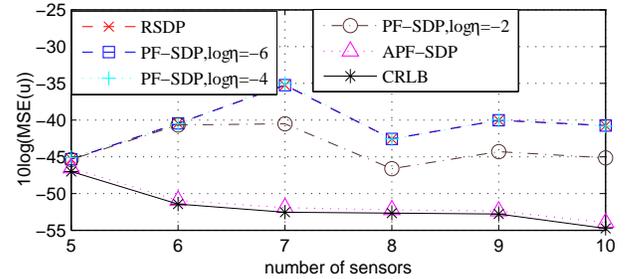}}
\subfigure[Performance of predicted velocity.]{
\label{Fig6.sub.2}
\includegraphics[width=0.48\textwidth]{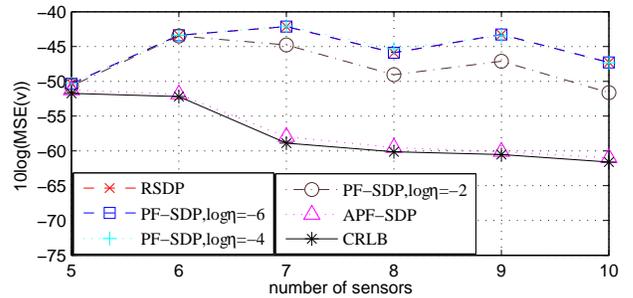}}
\caption{Performance comparison as number of sensors varies.}
\label{Fig6.lable}
\end{figure}
   The parameters in APF-SDP are listed as follows, $\alpha=5$, $\gamma = 10$, and $\delta = 10^{-5}$.  The first eight sensors listed in Tab.~\ref{tab:1} are used to locate the mobile source which starts from the initial position (0.2, -0.4) km with a constant velocity (-1, 1) m/s. When the noise level 10log$\sigma^2$ is increased from -60 dB to 20 dB, Fig.~\ref{Fig4.sub.1} and Fig.~\ref{Fig4.sub.2} illustrates the logarithmic MSEs of predicted initial position and velocity, respectively. It is obviously shown that the logarithmic MSE performance becomes worse as 10log$\sigma^2$ increases. When log$\eta$ is set to -6 and -4, the logarithmic MSEs of PF-SDP are almost same with  that of the RSDP, indicating the no working of penalty terms. However, the performance of PF-SDP becomes better when log$\eta$ is set to -2. Especially, the logarithmic MSE of PF-SDP is closer to the CRLB at large noise level, since the rank of PF-SDP solution is easier to be one. Compared with the RSDP, the AFP-SDP provides better performance in reaching the CRLB accuracy. Hence, it confirms the advantage of the APF-SDP in the position and velocity prediction.
\begin{figure*}[htb]
\centering
\subfigure[The scene in real experiments.]{
\label{Fig7.sub.1}
\includegraphics[width=0.32\textwidth]{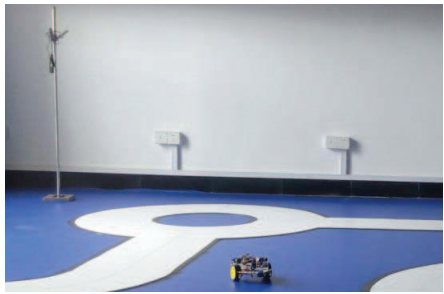}}
\subfigure[CDF of position error]{
\label{Fig7.sub.2}
\includegraphics[width=0.30\textwidth]{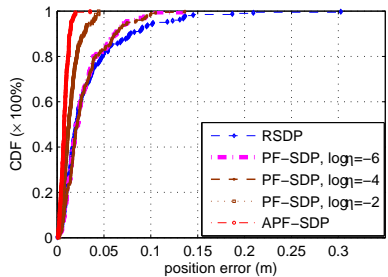}}
\subfigure[CDF of velocity error.]{
\label{Fig7.sub.3}
\includegraphics[width=0.30\textwidth]{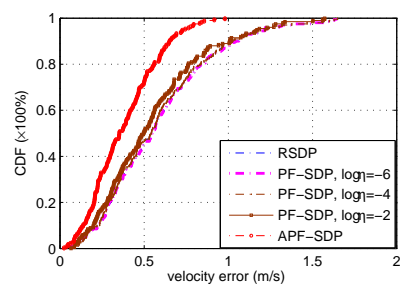}}
\caption{Performance comparison in real experiments.}
\label{Fig7.lable}
\end{figure*}

   The noise level 10log$\sigma^2$ is set to -40 dB and the other parameters are the same with those in Fig. 4. Fig.~\ref{Fig5.sub.1} and Fig.~\ref{Fig5.sub.2} also shows the logarithmic MSEs of predicted initial position and velocity when the propagation speed $c$ is increased from 0.1 km/s to 1 km/s. It is shown that the performance of predicted position or velocity becomes worse as the propagation speed increases. The performance of PF-SDP is always better that of the RSDP when log$\eta$ is set to -2, indicating the working of penalty terms. The APF-SDP always provides the almost CRLB performance in the prediction of position and velocity when $c$ is varied from 0.1 km/s to 1km/s. It confirms the advantages of AFP-SDP which can adaptively choose the penalty coefficient and provides optimal performance at the entire range of $c$.

  Finally, we also examine the impact of sensors on the performance of these proposed algorithms when the sensors are selected from Tab.~\ref{tab:1} in the order of list. The other parameter setup is also the same with that in Fig.4. Fig.~\ref{Fig6.sub.1} and Fig.~\ref{Fig6.sub.2} show the logarithmic MSEs of predicted initial position and velocity as the number of sensors varies. As can be seen that the logarithmic MSEs of PF-SDP (log$\eta$ = -2) are the same with those of RSDP at  when there are only five or six sensors. However, the PF-SDP (log$\eta$ = -2) provides better performance compared with the RSDP when the number of sensors is larger than seven. The proposed APF-SDP always reach the CRLB performance. It is also shown that adding more sensors can reduce the logarithmic MSE. However, how much reduction in logarithmic MSE by using additional sensors depends on the positions of the new sensors included.
\subsection{Real Experiments}
To verify the performance of our proposed algorithms, we also conducted the real experiments using a mobile car equipped with Ultrasonic module (UM). Besides the UM, motion sensors are also equipped to the mobile car and used to measure the parameters including the velocity and moving direction. Nine sensors are placed at the positions listed in Tab.~\ref{tab:4}. The UM transmits the ultrasonic signals to the sensors at the initial position, then the signals are received by the sensors and reflected to the UM of mobile car. Extensive tests show that the measurement noise in the range-difference is sufficiently close to be zero-mean Gaussian distribution with an predicted noise variance $6.3\times10^{-3}$ ms$^2$. The initial position of the mobile car is set to (1, 2) m, and the velocity of the mobile car is in the range of (1, 3) m/s.

 \begin{table}
    \centering
    \fontsize{8}{11}\selectfont
    \caption{Positions of nine sensors in real experiment(m)}
    \begin{tabular}{c|c|c|c|c|c|c|c|c|c}
    \toprule
    \hline
     $x$ & 5 & 5 & -5 & -5 & 5 & -5 & 0&  0&  0  \\\hline
    $y$ & 5 & -5 & 5 & -5 & 0 & 0 & 5 & -5&  0 \\\hline
     \bottomrule
    \end{tabular}\vspace{0cm}
    \label{tab:4}
\end{table}
\begin{figure}
 \centering
\subfigure[RMSE performance of predicted position.]{
\label{Fig8.sub.1}
\includegraphics[width=0.48\textwidth]{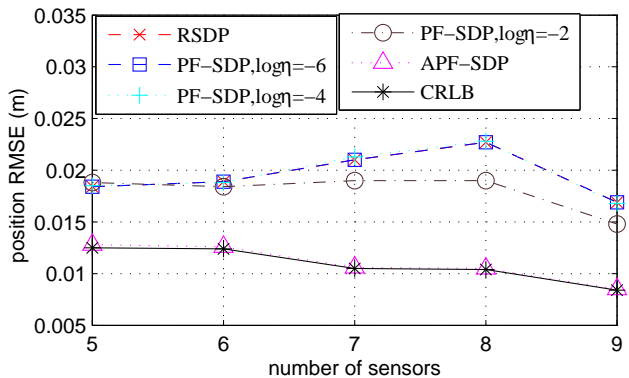}}
\subfigure[RMSE performance of predicted velocity.]{
\label{Fig8.sub.2}
\includegraphics[width=0.48\textwidth]{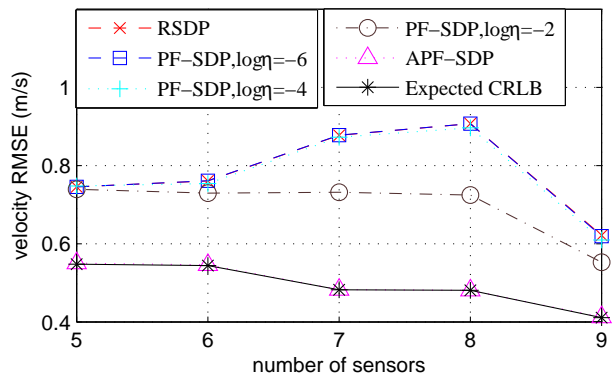}}
\caption{Performance comparison using experimental data as number of sensors varies.}
\label{Fig8.lable}
\end{figure}
We collect 200 sampling data that is used to predict the initial position and and velocity of the mobile car. Fig.~\ref{Fig7.sub.2} illustrates the cumulative distribution function (CDF) of position error with 200 runs with these different algorithms. As can be seen that 20\% of the position error is larger than 0.05 m for RSDP. However, it is reduced to near 0.01 m for the APF-SDP, indicating the advantage of our APF-SDP in the performance of position prediction. Fig.~\ref{Fig7.sub.3} shows the CDF of velocity error using 200 sampling data. 20\% of the velocity error is larger than 0.08 m/s for RSDP, but it is reduced to 0.06 m/s for our proposed APF-SDP. Therefore, the APF-SDP performs better than the RSDP in the prediction of velocity, confirming the advantage of the APF-SDP by achieving rank-one constraint.

Selecting the sensors from Tab.~\ref{tab:4} in the order of list, we also verify the performance of these proposed algorithm when the number of sensors is increased from 5 to 9. To clearly illustrate the effect of these different algorithms, we use the root mean square error (RMSE) to evaluate the performance with 200 sampling data. The position RMSE performance of these algorithms is plotted in Fig.~\ref{Fig8.sub.1}, where ``CRLB expected" represents the best achievable performance expected using the predicted noise variance. The RSDP performs not well even if the number of sensors increases. However, the proposed APF-SDP almost reaches the ``CRLB expected" performance, which is consistent with the simulation results. Fig.~\ref{Fig8.sub.2} illustrates the velocity RMSE as the number of sensors increases. The performance of PF-SDP is better than that of the RSDP when log$\eta$ is set to -2. However, the PF-SDP can not provide the ``CRLB expected" performance since it can not adaptively choose the penalty coefficient.
\subsection{Industrial Applications}
There are many position and velocity prediction problems in industrial Internet of things, such as the position obtaining of mobile AV and the tracking of UAV,
which are illustrated in Fig.~\ref{Fig9.sub.1} and Fig.~\ref{Fig9.sub.2}, respectively. Most of the current research on these problems is focused on the position obtaining of mobile source using some ranging information or the velocity prediction using motion sensors or Dropper shift measurements. The motion sensor is an extra hardware device which will increase the system cost. Since the velocity obtaining of Dropper shift measurement is subject to the error of (2, 4) m/s, it is considered to be unacceptable in some application systems. Our proposed method does not require any motion sensors or Dropper shift measurements and realize the prediction of position along with the velocity of mobile source.

Our experimental results show that the mean position RMSE of APF-SDP is smaller than 0.01 m when the UM equipped to the mobile source is used to collect the TD information. Apparently, the signal can be acoustic or electromagnetic, and the signal can travels in the underwater or underground scenarios. Due to the different propagation speed of the signal, our proposed algorithms provides distinct accuracy performance in the prediction of position and velocity, which is also demonstrated  in Fig. 5. Therefore, a feasible method to improve the accuracy performance is to reduce the noise level at large propagation speed. Moreover, the precise timing is very important to ensure the performance especially when the propagation speed of electromagnetic signal reaches $3\times10^8$ m/s.
\begin{figure}
 \centering
\subfigure[2-D scenario of mobile AV.]{
\label{Fig9.sub.1}
\includegraphics[width=0.32\textwidth]{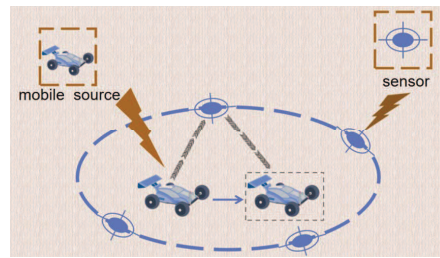}}
\subfigure[3-D scenario of mobile UAV.]{
\label{Fig9.sub.2}
\includegraphics[width=0.32\textwidth]{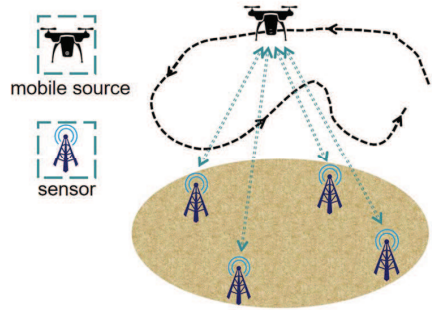}}
\caption{Different scenarios of our proposed system in industrial applications.}
\label{Fig9.lable}
\end{figure}

\section{Conclusions and Future Work}\label{sec:conclusion}
We introduce an intelligent practice prediction system for mobile source using the TD measurements. To predict the position and velocity of mobile source,
the RSDP algorithm is firstly proposed to obtain a convex SDP problem by dropping the rank-one constraint. The performance of RSDP is very poor due to the drop of rank-one constraint. Then the PF-SDP algorithm is put forward to obtain better performance by introducing the penalty terms. In the PF-SDP, the penalty coefficient is crucial to ensure the well performance of the PF-SDP and its optimal value is variable. Therefore, the AFP-SDP algorithm is proposed by adaptively choosing the penalty coefficient. Compared with the RSDP and PF-SDP, the AFP-SDP shows its performance advantages in the prediction of position and velocity. We have also done the simulations and the real experiments for the 2-D case. For the future work, our proposed methods are not applied in the 3-D scenario. Therefore, we will focus the future work on the position prediction of UAV and extend the system to the 3-D case.
\ifCLASSOPTIONcaptionsoff
  \newpage
\fi

\bibliographystyle{IEEEtran}
\small
\bibliography{dplc}

\end{document}